\documentclass[twocolumn,conference]{IEEEtran}
\usepackage{latexsym}
\usepackage{bm}
\usepackage[dvips]{graphics}
\usepackage{graphicx,subfigure}
\usepackage{psfrag}
\usepackage{amsfonts,amsmath,amssymb,hyperref,amsthm}
\usepackage{algorithm,algorithmic}

\usepackage{dsfont,color,subfigure}

\usepackage{xspace}
\usepackage{bbm}

\def\In{\mathop{\rm In}\nolimits}%
\def\Out{\mathop{\rm Out}\nolimits}%

\newcommand{\Ac}{\mathcal{A}}

\newcommand{\Ec}{\mathcal{E}}
\newcommand{\Fc}{\mathcal{F}}

\newcommand{\Mc}{\mathcal{M}}
\newcommand{\Nc}{\mathcal{N}}

\newcommand{\Rc}{\mathcal{R}}
\newcommand{\Sc}{\mathcal{S}}

\newcommand{\Vc}{\mathcal{V}}
\newcommand{\Wc}{\mathcal{W}}
\newcommand{\Xc}{\mathcal{X}}
\newcommand{\Yc}{\mathcal{Y}}

\newcommand{\Xv}{{\bf X}}

\newcommand{\xv}{{\bf x}}

\newcommand{\rv}{{\bf r}}

\newcommand{\Wv}{{\bf W}}
\newcommand{\Mv}{{\bf M}}

\newcommand{\Mh}{{\hat{M}}}

\def\a{\alpha}
\def\b{\beta}

\def\d{\delta}
\def\e{\epsilon}

\let\P\relax
\DeclareMathOperator\P{P}

\newcommand\ie{i.e.,\xspace}
\def\textiid{i.i.d.\@\xspace}
\newcommand\iid{\ifmmode\text{ i.i.d. } \else \textiid \fi}

\newtheorem{theorem}{Theorem}

\newtheorem{corollary}{Corollary}

\begin{document}

\title{	On the Impact of a Single Edge on the Network Coding Capacity}

\author{\IEEEauthorblockN{Shirin Jalali}
\IEEEauthorblockA{Center for Mathematics of Information\\
California Institute of Technology\\
shirin@caltech.edu}\and
\IEEEauthorblockN{Michelle Effros}
\IEEEauthorblockA{Departments of Electrical Engineering\\
California Institute of Technology\\
effros@caltech.edu} \and
\IEEEauthorblockN{Tracey Ho}
\IEEEauthorblockA{Department of Electrical Engineering\\
California Institute of Technology\\
tho@caltech.edu} 
}

\maketitle

\newcommand{\p}{\mathds{P}}
\newcommand{\mb}{\mathbf{m}}
\newcommand{\bb}{\mathbf{b}}

\begin{abstract}
In this paper, we study the effect of a single link on the capacity of a network of error-free bit pipes.  More precisely, we study the change in network capacity that results 
when we remove a single link of capacity $\delta$. 
In a recent result, we proved that if all the sources are directly available to a single super-source node, then removing a link of capacity $\delta$ cannot change the capacity region of the network by more than $\delta$ in each dimension. 
In this paper, we extend this result to
the case of multi-source, multi-sink networks
for some special network topologies.

\end{abstract}

\section{Problem Statement}

Consider a communication problem 
defined by a network, 
a collection of sources,  and 
a collection of sinks.  
The network is a directed graph with nodes 
representing communication devices 
and edges representing 
error-free, point-to-point communication channels 
with finite capacities.  
The sources are independent data streams, 
and each is available to precisely one node in the network.  
Each sink is a node in the network 
that desires some subset of the data streams; 
the desired subset may differ from one sink to the next.  
The capacity of the network, 
also called the ``network coding capacity,'' 
describes the set of achievable rates 
for every possible combination of sources and sinks. 
Solving for the capacity is a challenging open problem. 
In this paper, we investigate a simpler question: 
what is the effect of a single link 
on the network coding capacity of such a network? 
Specifically, 
we wish to understand whether decreasing the capacity 
of a single edge $e$ from $C_e\geq\d$ to $C_e-\d$ 
can change the capacity region of the network 
by more than $\d$ in each dimension.

In \cite{EffrosH:10}, we posed this question 
and proved that if all sources are available at one node, 
then changing the capacity of a single link by $\d$ 
reduces each achievable rate vector by at most $\d$ 
in each dimension. 
In this paper, we extend this result 
to a family of multi-source, multi-sink networks.  

\section{Notation}

Throughout the paper, finite sets are denoted by script letters such as $\Xc$ and $\Yc$. The size of a finite set $\Ac$ is denoted by $|\Ac|$. Random variables are denoted by upper case letters such as $X$ and $Y$. We represent the alphabet of  random variable $X$  by $\Xc$.  Bold letters, for example $\Xv=(X_1,\ldots,X_n)$ and $\xv=(x_1,\ldots,x_n)$ represent vectors. The length of a vector is implied in the context, and its $\ell^{\rm th}$ element is denoted by $X_{\ell}$.  For a set $\Fc\subseteq\{1,2,\ldots,n\}$, $\xv_{\Fc}=(x_i)_{i\in\Fc}$, where the elements are sorted in ascending order of their indices. For a vector $\Xv\in\mathds{R}^n$, let $\Xv^+=\max(\mathbf{0},\Xv)$, where $\mathbf{0}$ is a zero-valued vector of length $n$, and the $\max$ operator is applied component-wise.

\section{System Model}\label{sec:model}

Consider an acyclic error-free network $\Nc$ 
denoted by a directed graph $G=(\Vc,\Ec)$
with nodes $\Vc$ and edges $\Ec\subseteq\Vc\times\Vc$.  
Each edge $e=(v_1,v_2)\in\Ec$  represents 
an error-free channel from Node $v_1$ to Node $v_2$.  
We use $C_e>0$ to denote that channel's capacity.  
For each node $v\in\Vc$, 
$\In(v)=\{(v_1,v):(v_1,v)\in\Ec\}$ and 
$\Out(v)=\{(v,v_1):(v,v_1)\in\Ec\}$ 
denote the set of incoming and outgoing edges 
for Node $v$ respectively. 

Let $\Sc=\{1,2,\ldots,k\}$ denote the set of sources 
available in the network, 
and let  
\[
\a:\Sc\to\Vc,
\] 
specify the source availability.  
Thus for each $s\in\Sc$, 
$\a(s)$ describes the unique node where source $s$ is available.  
Likewise, for each $v\in\Vc$, 
let $\sigma(v)\subseteq\Sc$  denote the set of sources 
observed by Node $v$, \ie \[\sigma(v)=\{s: \a(s)=v\}.\] 
Finally, for each $v\in\Vc$, let $\b(v)\subseteq\Sc$ 
denote the set of sources  that Node $v$ is interested in recovering.

A network code of block length $n$ 
and rate $\mathbf{R}=(R_s)_{s\in\Sc}$ 
over such a network is described as follows. 
Each source $s\in\Sc$ generates some message 
$M_s\in\Mc_s=\{1,2,\ldots,2^{nR_s}\}$. 
For each $e\in\Ec$, let $\Wc_e=\{1,2,\ldots,2^{nC_e}\}$. 
The coding operations performed by each node can be categorized as follows
\begin{enumerate}
\item Encoding functions:\\  For each $v\in\Vc$ and $e\in\Out(v)$, 
the encoding function corresponding to Edge $e$ is a mapping
\begin{align}
g_e:&\prod_{s\in\sigma(v)}\Mc_{s}
\times\prod\limits_{e'\in\In(v)}\Wc_{e'}\to\Wc_e.\nonumber
\end{align}
\item Decoding functions:\\ For each $v\in\Vc$ and $s\in\b(v)$, 
the decoding function for source  $s$ at Node $v$ is a mapping
\begin{align}
g_v^{s}:\prod_{s'\in\sigma(v)}\Mc_{s'}\times\prod\limits_{e\in\In(v)}
\Wc_e\to\Mc_s.\nonumber
\end{align}
\end{enumerate}
A rate vector $\mathbf{R}=(R_s)_{s\in\Sc}$ is said to be achievable on network $\Nc$, if for any $\e>0$, there exists a block length  $n$ large enough and a coding scheme of block length $n$ operating at rate $\mathbf{R}$ such that for all $v\in\Vc$ and $s\in\b(v)$
\begin{align}
\P(\hat{M}_s^{(v)}\neq M_s )\leq \e,\nonumber
\end{align}
where $\hat{M}_s^{(v)}$ denotes the reconstruction 
of message $M_s$ at Node $v$. 
For sources $\Sc$, availability mapping $\a(\cdot)$, 
and demand mapping $\b(\cdot)$, 
let $\Rc(\Nc,\Sc,\a,\b)$ denote the set of achievable rates 
on Network $\Nc$.

In the discussion that follows, 
we use $\Nc$ to describe the original network 
and $\Nc'$ to describe the new network that results 
when we reduce the capacity of a single, fixed edge $e\in\Ec$ 
from $C_e\geq\d$ to $C'_e=C_e-\d$.  
If $C_e=\d$, then edge $e$ is removed from $\Nc$ to obtain $\Nc'$.

\section{Prior Work}\label{sec:prior}

Network codes are communication schemes in which every node is allowed to perform arbitrary functions on its inputs in creating its outputs. The idea was first proposed by Ahlswede, Cai, Li, and Yeung in 2000~\cite{Ahlswede:00}. 
They proved that  Ford and Fulkerson's famous max-flow min-cut theorem for unicast networks \cite{FordF:56}, also holds in multicast networks. (Here a ``unicast network'' refers to a network with a single source and a single sink node, 
while a ``multicast network'' refers to a network with one source 
and multiple sink nodes, 
each requiring all data available at the source.) 
While it is always possible to achieve the capacity in a unicast network using only routing at the relay nodes, Ahlswede et al. showed that there exist networks where coding is required to achieve the multicast capacity.  
 Linear coding operations suffice for achieving the capacity of a multicast network by \cite{LiY:2003}. While both  the capacity region and the structure of capacity-achieving codes  are known for multicast demands, 
neither the capacity nor a low-complexity family of codes 
sufficient for achieving the capacity 
is known for most demand types. 
Linear codes are insufficient 
for achieving the capacity under general demands by \cite{DoughertyF:05}. 
%

Computing the capacity region of an error-free network can be cast as a convex optimization problem with a linear cost function over the space of \emph{normalized entropic vectors} with some other linear constraints \cite{HassibiS:07}\cite{YanY:07}. This characterization reveals that network information theory problems over noiseless networks could be solved if we could explicitly characterize the set of entropy vectors. While there has been a lot of effort in recent years 
geared towards developing a better understanding 
of the set  of entropy vectors 
(c.f. \cite{Yeung:97,ZhangY:97, ChanY:99, ChanY:02,DoughertyF:06,Matus:07}), 
to date the problem remains largely unsolved.

In this paper, we study the problem from a different perspective. Instead of trying to find the capacity region of a network, we focus on the effect 
of a single link on that capacity region. 
Precisely, we try to understand the effect 
on network capacity 
of changing the capacity on a single edge $e\in\Ec$ 
from $C_e\geq \d$ to $C_e'=C_e-\d$, 
which effectively changes just one linear constraint 
in the problem as described above.  

\section{Results}

Before stating our main result  in Section \ref{sec:results}, 
we briefly review some cases where the impact, 
in terms of network capacity, 
of reducing $C_e$ is already known or straightforward to characterize.  

\subsection{Demand Types with Tight Cut-Set Bounds}

For a variety of demand types, including multicast, multi-source multicast,  single-source with non-overlapping demands, and   single-source with non-overlapping demands and a multicast demand, network coding capacity can be fully characterized by the corresponding cut-set bounds~\cite{KoetterM:03}. 
Reducing $C_e$ to $C_e-\delta$ for a single edge $e\in\Ec$ 
reduces the capacity of every cut by at most $\d$. 
Therefore, if $(\Sc,\a,\b)$ describes any such demand type, 
and $\mathbf{R}\in\Rc(\Nc,\Sc,\a,\b)$, 
then $(\mathbf{R}-\d\cdot\mathbf{1})^+\in\Rc(\Nc',\Sc,\a,\b)$, 
where $\Nc'$ is the modified network, 
as described in Section~\ref{sec:model}, and 
$\mathbf{1}$ is the all-ones vector.

\subsection{Links Connected to Terminal Nodes}

Consider a terminal node $v_o\in\Vc$;  
then Node~$v_o$ has no outgoing edges ($\Out(v_o)=\emptyset$).  
Let $p=|\In(v_o)|$ denote the number of edges incoming to $v_o$, 
and let $W_1,W_2,\ldots,W_p$ 
denote the messages carried by these links. 
Further, assume that the link corresponding 
to the message $W_1$ has capacity $\d$. 
For any $s\in\b(v_o)$, 
\begin{align}
&I(M_s;W_2,\ldots,W_p)\nonumber\\
&=I(M_s;W_1,W_2,\ldots,W_p)-I(M_s;W_1|W_2,\ldots,W_p)\nonumber\\
&\geq I(M_s;W_1,W_2,\ldots,W_p)-H(W_1)\nonumber\\
&\geq I(M_s;W_1,W_2,\ldots,W_p)-n\d.\nonumber
\end{align}
This proves that removing this link 
reduces the capacity from source $s$ 
to node $v$ by at most $\d$. 
Since Node $v$ has only incoming edges, 
this change does not affect the capacities 
at any other nodes in the network.  
As a result, 
applying, for each $s\in\sigma(v)$, 
an outer code with rate $R_s-\d$ 
and codewords drawn uniformly at random 
yields expected error probability 
approaching 0 as the coding dimension grows without bound.  
This proves the existence of a good collection of codes.  
Therefore, $\mathbf{R}=(R_s:s\in\Sc)\in\Rc(\Nc,\Sc,\a,\b)$, 
implies $\mathbf{R}'=(R_s':s\in\Sc)\in\Rc(\Nc',\Sc,\a,\b)$, 
where $R_s'=R_s$ for all $s\in\Sc\setminus\sigma(v)$ 
and   $R_s'=(R_s-\d)^+$ for all $s\in\sigma(v)$.  

\subsection{Super Source Node}

For the case where all the sources 
are available to a \emph{ super source node} 
($\sigma(v_o)=\Sc$ for some $v_o\in\Vc$, 
as shown in Fig.~\ref{fig:net_colloc}), 
we showed in~\cite{EffrosH:10} 
that changing the capacity of any link $e\in\Ec$ 
from $C_e\geq\d$ to $C_e'=C_e-\d$ 
changes the network capacity region by at most $\d$ in each dimension 
(i.e., $\mathbf{R}\in\Rc(\Nc,\Sc,\a,\b)$ 
implies $(\mathbf{R}-\d\cdot\mathbf{1})^+\in\Rc(\Nc',\Sc,\a,\b)$.  

\begin{figure}[h]
\begin{center}
\psfrag{s}[c]{$v_o$}
\psfrag{C1}[c]{$\d$}
\psfrag{C2}[l]{$C$}
\psfrag{N}[l]{$\Nc$}
\psfrag{M1}[l]{$M_1$}
\psfrag{M2}[l]{$M_2$}
\psfrag{Mk}[l]{$M_k$}
\psfrag{Mh1}[l]{$\Mh_1$}
\psfrag{Mh2}[l]{$\Mh_2$}
\psfrag{Mhk}[l]{$\Mh_k$}
\psfrag{dots}[l]{$\vdots$}
\includegraphics[width=7cm]{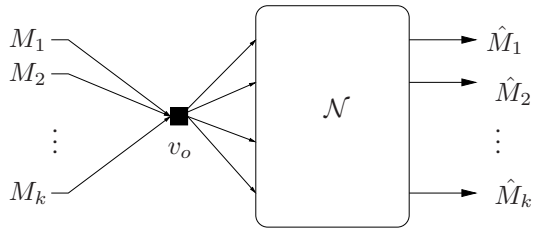}
\caption{All sources available directly at a super source node $v_o$}
\label{fig:net_colloc}
\end{center}
\end{figure}

\subsection{Linear Network Coding}

Consider a linear network code of block length $n$ 
and rate $\mathbf{R}=(R_s)_{s\in\Sc}$ 
operating on network $\Nc$. 
Let $e\in\Ec$ be a fixed link of capacity $C_e=\d$ inside this network. 
In this case, we treat both source messages 
and the messages traversing each link in the network 
as binary vectors.  
Since the code is linear, 
the message $W_e$ sent across link $e$ 
can be written as a linear combination 
of the source messages $\{M_s\}_{s\in\Sc}$.  
Precisely, 
\begin{align}
W_e=\sum\limits_{s\in\Sc}A_{s,e}M_s,\label{eq:linear}
\end{align} 
where for each $s\in\Sc$, $A_{s,e}$ 
denotes a binary matrix 
of dimension $nC_e\times nR_s$ 
and all additions in \eqref{eq:linear} are binary additions.  
Let $\Mc_0$ denote the set of messages 
that yield message $W_e=\mathbf{0}$ on link $e$ 
using the given linear code, \ie 
\[
\Mc_0\triangleq\{(M_s)_{s\in\Sc}:\sum\limits_{s\in\Sc}A_{s,e}M_s
=\mathbf{0}\}.
\]
If we restrict our attention to this subset of messages, 
then we can run the given linear code 
in the absence of edge $e$ 
since the value of $W_e$ for all such messages 
is fixed and known.  
Unfortunately, 
choosing messages from $\Mc_0$ 
may require coordination among the source nodes.  
We therefore choose messages from a subset of $\Mc_0$ 
that requires no such coordination.  
Namely, we transmit only messages from $\Mc_{00}$, 
where $\Mc_{00}$ is defined as 
\[
\Mc_{00}\triangleq\{(M_s)_{s\in\Sc}:A_{s,e}M_s=\mathbf{0}
\mbox{ for all }s\in\Sc\}.  
\]
By sending only messages $(M_s)_{s\in\Sc}\in\Mc_{00}$, 
we guarantee that $W_{e}={\mathbf{0}}$;  
since $\Mc_{00}=\prod_{s\in\Sc}\{M_s:A_sM_s=0\}$, 
the source nodes can transmit only messages from $\Mc_{00}$ 
without coordination.  
The resulting rate is 
$(1/n)\log|\{M_s:A_sM_s=\mathbf{0}\}|\geq(R_s-\d)^+$ 
for each $s\in\Sc$.   
Thus we can apply the code from $\Nc$ 
on the network $\Nc'$ 
to achieve reliable communication at rate 
$(\mathbf{R}-\d\cdot\mathbf{1})^+$.  

The given argument demonstrates that removing a single link 
of capacity $C_e=\d$ changes the rate achievable 
with linear coding by at most $\d$ in each dimension.  
The same argument can be used to 
show that reducing the capacity of some edge $e$ 
with $C_e>\d$ to $C_e'=C_e-\d$
reduces the rate achievable 
with linear coding by at most $\d$ in each dimension.  
This can be seen by treating a link of capacity $C_e>\d$ 
as a pair of parallel links 
of capacities $C_e-\d$ and $\d$, respectively, 
and applying the previous argument.  

Unfortunately, as noted in Section~\ref{sec:prior}, 
linear network codes are not sufficient 
for achieving the capacity of general error-free networks. 
Thus, the given strategy 
proves only that reducing the capacity of a link by $\d$ 
changes the set of rates achievable using linear coding 
by at most $\d$ in each dimension.  
If rate $\mathbf{R}$ is achievable using linear coding on $\Nc$, 
then rate $(\mathbf{R}-\d\cdot\mathbf{1})^+$ is achievable 
using linear coding on $\Nc'$.  
 
\subsection{Main Result}\label{sec:results}


Consider the $k$-unicast network $\Nc$ shown in Fig.~\ref{fig:net_2link}. 
Here, $\a(s)=v_s$ and $\b(v_{k+s})=\{s\}$
for all $s\in\Sc$; 
that is, each message $s\in\Sc$ is a unicast from node $v_s$ 
to node $v_{k+s}$.  
In a blocklength-$n$ code, 
$\Mc_s\in\{1,2,\ldots,2^{nR_s}\}$ 
denotes the source message for Source $s$, 
and $\Mh_s$ represents the reconstruction of $M_s$ 
at sink node $v_{k+s}$.  
When we remove the link $e$ of capacity $C_e=\d$ from $\Nc$, 
we obtain the network  $\Nc'$ shown in Fig.~\ref{fig:net_1link}.

\begin{figure}[h]
\begin{center}
\psfrag{a}[c]{$a$}
\psfrag{b}[l]{$b$}
\psfrag{c}[l]{}
\psfrag{d}[l]{}
\psfrag{C1}[c]{$\d$}
\psfrag{C2}[l]{$C$}
\psfrag{N1}[l]{$\Nc_1$}
\psfrag{N2}[l]{$\Nc_2$}
\psfrag{M1}[l]{$M_1$}
\psfrag{M2}[l]{$M_2$}
\psfrag{Mk}[l]{$M_k$}
\psfrag{Mh1}[l]{$\Mh_1$}
\psfrag{Mh2}[l]{$\Mh_2$}
\psfrag{Mhk}[l]{$\Mh_k$}
\psfrag{dots}[l]{$\vdots$}
\subfigure[Network $\Nc$]{\includegraphics[width=8.5cm]{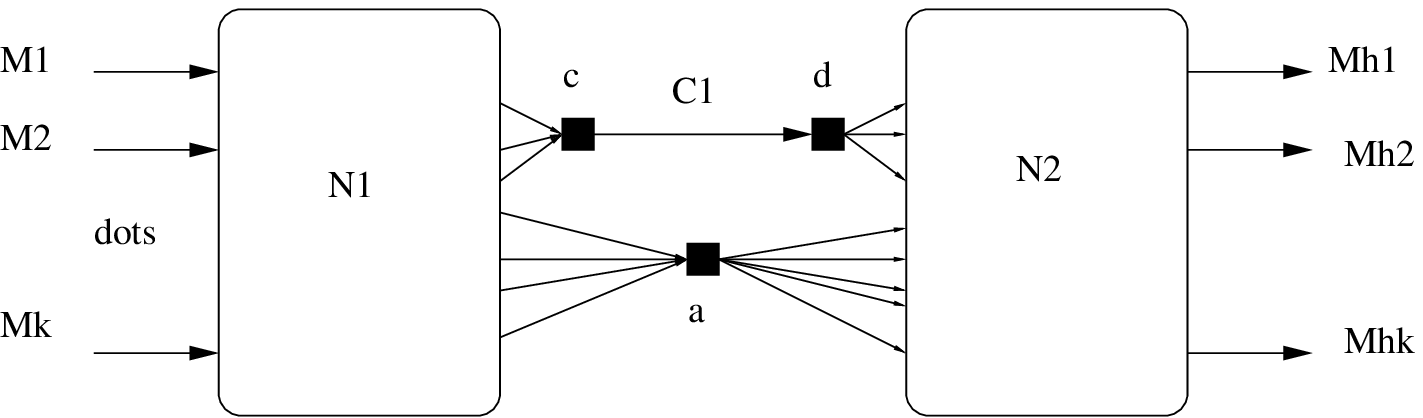}\label{fig:net_2link}}
\subfigure[Network $\Nc'$]{\includegraphics[width=8.5cm]{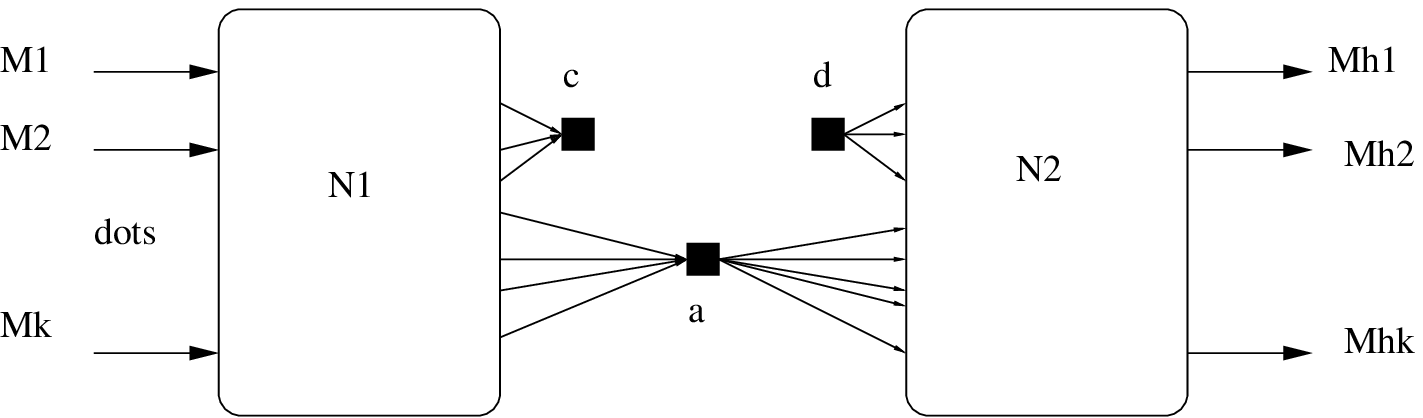}\label{fig:net_1link}}
\caption{A multiple unicast network with special structure} 
\end{center}
\end{figure}
\begin{theorem}\label{thm:main}
For any $\mathbf{R}\in\Rc(\Nc,\Sc,\a,\b)$, 
\[
(\mathbf{R}-\d\cdot\mathbf{1})^+\in\Rc(\Nc',\Sc,\a,\b).
\] 
\end{theorem}
\begin{proof}
Fix $\mathbf{R}=(R_1,R_2,\ldots,R_k)\in\Rc(\Nc,\Sc,\a,\b)$.  
We first consider the case where $\min\{R_1,\ldots,R_k\}\geq \d$.
Given a code of blocklength $n$, 
for each $s\in\Sc$, 
let $P_{e,s}^{(n)}\triangleq\P(M_s\neq \Mh_s)$ 
denote the error probability in reconstructing source $s$ at sink $v_{k+s}$. 
For any $p\in[0,1]$, 
let $h(p)=-p\log(p)-(1-p)\log(1-p)$ be the binary entropy function. 
Since $\mathbf{R}$ is  achievable on $\Nc$, 
for any $\e>0$ and $n$ large enough 
there exists a rate-$\mathbf{R}$ code of blocklength $n$ 
such that $\max\{P_{e,s}^{(n)}:s\in\Sc\}\leq \e$ 
and $\max\{h(P_{e,s}^{(n)}):s\in\Sc\}\leq \e$.  
Given any $\e>0$, fix such a code.  
We next use this family of codes to prove the existence 
of a multiple access code for communicating the sources 
from nodes $v_1,\ldots,v_k$ to node $a$ 
and a broadcast code for transmitting all sources $s\in\Sc$ 
from Node $a$ to nodes $v_{k+1},\ldots,v_{2k}$, respectively, 
both at rates $\mathbf{R}-\d\cdot\mathbf{1}$.  
In the arguments that follows, 
we use $W_e$, $\Wv_i$, and $\Wv_o$ to denote 
the message sent through the link $e$ of capacity $C_e=\d$, 
the inputs of Node $a$, and 
the outputs of Node $a$, respectively (see Fig.~\ref{fig:net_2link}).

Consider the $k$-user multiple access channel 
with inputs $\Mv=(M_1,M_2,\ldots,M_k)$ and output $\Wv_i$. 
The capacity region of this $k$-user MAC is the set of rate vectors 
$\rv=(r_1,r_2,\ldots,r_k)$ satisfying 
\begin{align}
\sum_{s\in\Ac}r_s\leq I(\Mv_{\Ac};\Wv_i|\Mv_{\Ac^c},Q),\nonumber
\end{align}
for all $\Ac\subseteq\Sc$ and some 
\[
p(q)p(m_1|q)p(m_2|q)\ldots p(m_k|q).
\]
Define  
\[
\rv_{\rm{mac}}\triangleq(I(M_1;\Wv_i),I(M_2,\Wv_i),\ldots,I(M_k;\Wv_i)),
\] 
under the distribution imposed by the code fixed above.  
In the argument that follows, we first show that $\rv_{\rm{mac}}$ 
falls in the capacity region of the MAC 
and then prove that $\rv_{\rm{mac}}$ satisfies the desired rate constraint.  

Since the messages $M_1,\ldots,M_k$ are independent, 
for any sets $\Ac\subseteq\Sc$ and $\Ac^c=\Sc\setminus\Ac$,
\begin{align}
\sum\limits_{s\in\Ac} r_{{\rm mac},s} &= \sum\limits_{s\in\Ac} I(M_s;\Wv_i)\nonumber\\
&= \sum\limits_{s\in\Ac} [H(M_s)-H(M_s|\Wv_i)]\nonumber\\
&=H(M_{\Ac})-\sum\limits_{s\in\Ac}H(M_s|\Wv_i)\nonumber\\
&\leq  H(M_{\Ac}) -H(M_{\Ac}|\Wv_i)\nonumber\\
&\leq  H(M_{\Ac}) -H(M_{\Ac}|\Wv_i,M_{\Ac^c})\nonumber\\
& =I(M_{\Ac};\Wv_i|M_{\Ac^c}).\nonumber
\end{align} 
Thus, $\rv_{\rm{mac}}$ falls in the capacity region of the MAC.  

We next bound each term in $\rv_{\rm{mac}}$.  
For each $s\in\Sc$, 
\begin{align}
H(M_s|\Wv_i)&\leq H(M_s,W_e|\Wv_i)\nonumber\\
&= H(M_s|W_e,\Wv_i)+H(W_e|\Wv_i)\nonumber\\
&\leq nR_sP_{e,s}^{(n)} +h(P_{e,s}^{(n)})+n\d \nonumber
\end{align}
by Fano's inequality~\cite{cover}.
Hence,
\begin{align} 
I(M_s;\Wv_i) &= H(M_s)-H(M_s|\Wv_i)\nonumber\\ 
&\geq n(R_s-\d) - nR_s\e-\e,\label{eq:mac}
\end{align}
since $\max\{P_{e,s}^{(n)},h(P_{e,s}^{(n)})\}\leq \e$ by assumption.
Recall that $\e>0$ is arbitrary;  
thus \eqref{eq:mac} implies that 
$(\mathbf{R}-\d\cdot\mathbf{1})$
is achievable on the described MAC. 

We next deliver these messages 
to their intended receivers 
using the broadcast channel (BC)  
from Node $a$ to the sinks  $v_{k+1},\ldots,v_{2k}$.  
Again, we apply the previously chosen code, 
operating the code in the absence of edge $e$ 
by sending only source messages for which the message across edge $e$ 
is a fixed value $w_e$ to be chosen next.  

Note that 
\begin{align}
H(\Mv)&=H(\hat{\Mv})+H(\Mv|\hat{\Mv})\nonumber
\end{align}
\begin{align}
&= H(\hat{\Mv})+ \sum_{s=1}^kH(M_s|M^{s-1},\hat{\Mv})\nonumber\\
&\leq H(\hat{\Mv})+ \sum_{s=1}^kH(M_s|\hat{M}_s)\nonumber\\
&\stackrel{(a)}{\leq} H(\hat{\Mv}) + \sum_{s=1}^k( h(P_{e,s}^{(n)})+ nR_s P_{e,s}^{(n)})\nonumber\\
&\stackrel{(b)}{\leq} H(\hat{\Mv}) +k\e+n\e \sum_{s=1}^k R_s.\label{eq:1}
\end{align}
where $\rm (a)$ and $\rm (b)$  follow from the Fano's inequality \cite{cover}, and our initial assumption, respectively. Hence, from \eqref{eq:1},
\begin{align}
H(\hat{\Mv}) &\geq  H({\Mv}) -k\e- n\e \sum_{s=1}^k R_s\nonumber\\
&= (1-\e )n\sum_{s=1}^k R_s-k\e.\label{eq:1-1}
\end{align}

On the other hand, we have
\begin{align}
H(\hat{\Mv}|W_e)&=H(\hat{\Mv},W_e)-H(W_e)\nonumber\\
&\geq H(\hat{\Mv})-H(W_e)\nonumber\\
&\geq H(\hat{\Mv})-n\d.\label{eq:2}
\end{align}
Therefore, combining \eqref{eq:1-1} and \eqref{eq:2}, it follows that 
\begin{align}
H(\hat{\Mv}|W_e) \geq (1-\e )n\sum_{s=1}^k R_s-k\e -n\d.\label{eq:3}
\end{align}
Since $H(\hat{\Mv}|W_e)=\sum_{w_e\in\Wc_e}H(\hat{\Mv}|W_e=w_e)p(w_e)$, there exists some $w_e\in\Wc_e$ such that 
\begin{align}
H(\hat{\Mv}|W_e=w_e)&\geq (1-\e )n\sum_{s=1}^k R_s-k\e -n\d.\label{eq:4}
\end{align}
%
%
Fixing the message $W_e$ to a value of $w_e$ 
that satisfies \eqref{eq:4}, 
we get a $k$-user deterministic broadcast channel (BC) \cite{cover} 
with input $\Wv_o$  and outputs $(\Mh_1,\ldots,\Mh_k)$. 
Appendix A summarizes prior results 
on the capacity region for this BC, 
which achieves reliable transmission 
at all rates $\rv=(r_1,r_2,\ldots,r_k)$ for which
\begin{align}
\sum_{s\in\Ac}r_s\leq H(\Mh_{\Ac}|W_e=w_e),\nonumber
\end{align}
for all $\Ac\subseteq\Sc$. 
We now prove that this set of rates 
includes the rate $\rv_{\rm bc}=n(\mathbf{R}-\d\cdot\mathbf{1})$. 
For any $\Ac\subseteq\Sc$, we have
\begin{align}
& H(\Mh_{\Ac}|W_e=w_e)+H(\Mh_{\Ac^c}|W_e=w_e)\nonumber\\
&\geq H(\hat{\Mv}|W_e=w_e).\label{eq:entropy_upper_bound_out}
\end{align} 
But $H(\Mh_{\Ac^c}|W_e=w_e)\leq \sum_{s\in\Ac^c}nR_s$. 
Hence, combining \eqref{eq:4} 
and \eqref{eq:entropy_upper_bound_out},
\begin{align}
H(\Mh_{\Ac}|W_e=w_e)\geq  n\sum_{s\in\Ac}R_s- n\sum_{s=1}^kR_s\e-k\e-n\d.\nonumber
\end{align} 
Thus, since $\e$ is arbitrary, 
$n(\mathbf{R}-\d\cdot\mathbf{1})$ is achievable 
on the given BC.  
This implies that the messages 
received by node $a$ at rate $\rv_{\rm mac}$ 
can be delivered to their intended receivers, 
which concludes the proof for the case where $R_s>\d$ for all $s\in\Sc$.

Finally, note that if there are some sources with $R_s\leq \d$, 
then we can use the same argument 
by sending constant messages for all such sources 
in both the MAC and the BC.  
\end{proof}

A special case of the network shown in Fig.~ \ref{fig:net_2link} 
is  shown in Fig.~\ref{fig:net_no_link}. 
Theorem \ref{thm:main} immediately applies.  

\begin{figure}[h]
\begin{center}
\psfrag{a}[c]{}
\psfrag{b}[l]{}
\psfrag{c}[l]{}
\psfrag{d}[l]{}
\psfrag{C1}[c]{$\d$}
\psfrag{C2}[l]{$C$}
\psfrag{N1}[l]{$\Nc_1$}
\psfrag{N2}[l]{$\Nc_2$}
\psfrag{M1}[l]{$M_1$}
\psfrag{M2}[l]{$M_2$}
\psfrag{Mk}[l]{$M_k$}
\psfrag{Mh1}[l]{$\Mh_1$}
\psfrag{Mh2}[l]{$\Mh_2$}
\psfrag{Mhk}[l]{$\Mh_k$}
\psfrag{dots}[l]{$\vdots$}
\includegraphics[width=8.5cm]{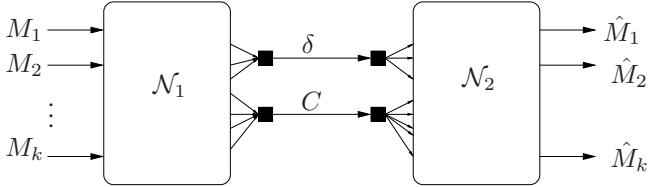}
\caption{A special case of the network shown in Fig. \ref{fig:net_2link}}\label{fig:net_no_link}
\end{center}
\end{figure}

Note that Theorem \ref{thm:main} can also be used to derive an outer bound on the capacity region of the $k$-unicast network $\Nc$ shown in Fig. \ref{fig:net_2link}. Let $\Rc_1\triangleq\Rc(\Nc_1,\Sc,\a_1,\b_1)$ and $\Rc_2\triangleq\Rc(\Nc_2,\Sc,\a_2,\b_2)$ denote the capacity regions of the networks $\Nc_1$ and $\Nc_2$ shown in Fig. \ref{fig:net_1link}, with $\a_1(s)=v_s$, $\a_2(s)=a$ and $\b_2(v_{s+k})=s$, for $s\in\Sc$. Moreover, $\b_1(a)=\Sc$, $\b_1(v)=\emptyset$ for $v\in\Vc\backslash a$, and  $\b_2(v)=\emptyset$ for $v\in\Vc\backslash\{v_{k+1},\ldots,v_{2k}\}$. Note that $\Rc_1$ and $\Rc_2$ correspond to a multicast network and a single source network with non-overlapping demands, respectively. Hence, as mentioned before, in both cases the capacity regions are computable and are fully characterized by the cut-set bounds \cite{KoetterM:03}. 
\begin{corollary} 
Let $\Rc_o\triangleq \{R+\d\cdot\mathbf{1}: R\in\Rc_1\cap\Rc_2 \}$. Then,
\[
\Rc(\Nc,\Sc,\a,\b)\subseteq\Rc_o.
\]
\end{corollary}

\section{Conclusion}\label{sec:conclusion}

In this paper we study the effect 
of a single link on the network coding capacity 
of a network of error-free bit pipes. 
For some special topologies of multi-source multi-sink networks, 
we prove that  our result from \cite{EffrosH:10} continues to hold; 
that is, reducing the capacity of a link by $\d$ 
changes the capacity region by at most $\d$ in each dimension. 
The question of whether or not this result holds 
for all networks remains an open area for future research.

\renewcommand{\theequation}{A-\arabic{equation}}
\setcounter{equation}{0}  

\section*{APPENDIX A\\ Deterministic broadcast channel}  \label{app1} 

A $k$-user deterministic broadcast channels (DBC) 
with input $x\in\Xc$ and outputs $\{Y_s\in\Yc_s\}_{s\in\Sc}$ 
is a $k$-user broadcast channel such that 
for any $x\in\Xc$ and $(y_1,\ldots,y_k)\in\Yc_1\times\Yc_2\times\ldots\times\Yc_k$,
\begin{align}
\P((Y_1,\ldots,Y_k)=(y_1,\ldots,y_k)|X=x)\in\{0,1\}.\label{eq:a_1}
\end{align}
Since the capacity region of a BC depends only 
on the receivers' conditional marginal distributions \cite{cover}, 
\eqref{eq:a_1} implies that a $K$-user DBC 
can be described by $k$ functions $(f_1,\ldots,f_k)$,  \[f_s:\Xc\to\Yc_s,\]
such that  $Y_s=f_s(X)$ for $s\in\Sc$.

While the capacity region for general BCs remains unsolved, 
the capacity region of a $k$-user DBC is known 
and can be described by the union of the set of 
rates $(R_1,R_2,\ldots,R_k)$ satisfying 
\begin{align}
\sum\limits_{s\in\mathcal{A} }R_s \leq H(Y_{\mathcal{A}}),\nonumber
\end{align}
for any $\mathcal{A}\subseteq\{1,\ldots,k\}$, for some $P(X)$  \cite{Marton:77,Pinsker:78} .

\section*{Acknowledgments}
This work was supported in part by 
Caltech's Center for the Mathematics of Information (CMI), 
DARPA ITMANET grant W911NF-07-1-0029, 
the Air Force Office of Scientific Research 
under grant FA9550-10-1-0166, 
and Caltech's Lee Center for Advanced Networking.

\bibliographystyle{unsrt}
\bibliography{myrefs}

\end{document}